\documentclass[nonacm]{aamas} 

\usepackage{balance} 
\usepackage{graphicx}
\usepackage{tikz}
\usepackage{framed}
\usepackage{amsmath}
\usepackage{booktabs}
\usepackage{algorithm}
\usepackage{algorithmic}

\setcopyright{ifaamas}
\acmConference[AAMAS '23]{Proc.\@ of the 22nd International Conference
on Autonomous Agents and Multiagent Systems (AAMAS 2023)}{May 29 -- June 2, 2023}
{London, United Kingdom}{A.~Ricci, W.~Yeoh, N.~Agmon, B.~An (eds.)}
\copyrightyear{2023}
\acmYear{2023}
\acmDOI{}
\acmPrice{}
\acmISBN{}

\acmSubmissionID{817}

\title[AAMAS-2023 Formatting Instructions]{Cost Sharing under Private Valuation and Connection Control}

\author{Tianyi Zhang$^*$}
\affiliation{
  \institution{ShanghaiTech University}
  \city{}
  \country{China}
  }
\email{zhangty@shanghaitech.edu.cn}

\author{Junyu Zhang$^*$}
\affiliation{
  \institution{ShanghaiTech University}
  \city{}
  \country{China}}
\email{zhangjy22022@shanghaitech.edu.cn}

\author{Sizhe Gu}
\affiliation{
  \institution{ShanghaiTech University}
  \city{}
  \country{China}}
\email{guszh@shanghaitech.edu.cn}

\author{Dengji Zhao}
\affiliation{
  \institution{ShanghaiTech University}
  \city{}
  \country{China}}
\email{zhaodj@shanghaitech.edu.cn}

\begin{abstract}
We consider a cost sharing problem on a weighted undirected graph, where all the nodes want to connect to a special node called source, and they need to share the total cost (weights) of the used edges. Each node except for the source has a private valuation of the connection, and it may block others' connections by strategically cutting its adjacent edges to reduce its cost share, which may increase the total cost. We aim to design mechanisms to prevent the nodes from misreporting their valuations and cutting their adjacent edges. We first show that it is impossible for such a mechanism to further satisfy budget balance (cover the total cost) and efficiency (maximize social welfare). Then, we design two feasible cost sharing mechanisms that incentivize each node to offer all its adjacent edges and truthfully report its valuation, and also satisfy either budget balance or efficiency.
\end{abstract}


\newcommand{\BibTeX}{\rm B\kern-.05em{\sc i\kern-.025em b}\kern-.08em\TeX}

\begin{document}


\pagestyle{fancy}
\fancyhead{}


\maketitle 
\renewcommand{\thefootnote}{*}
\footnotetext{The authors have equal contributions.}

\section{Introduction}
\label{intro}
In the classic cost sharing problem, there are a group of agents at different locations and a source. All the agents want to connect to the source via the connections (edges) between the locations, but each connection has a cost~\cite{claus1973cost,granot1981minimum,lorenzo2009characterization}. The goal is to allocate the total connection cost among the agents. This problem exists in many real-world applications such as cable TV, electricity, and water supply networks~\cite{bergantinos2010minimum,gomez2011merge,trudeau2017set}. It has been well-studied and many solutions have been proposed to achieve different properties~\cite{kar2002axiomatization,bergantinos2007optimistic,trudeau2012new} (we survey them in Section~\ref{anay}).  

However, these solutions do not consider two natural strategic behaviors of the agents. First, to connect to the source, an agent may need to go through some intermediate agents. These agents may block the connection by strategically cutting their adjacent edges if their cost share is reduced by doing so~\cite{zhao22mechanism} (see an example in Section~\ref{anay}), which will potentially increase the total cost of connecting the agents. Second, each agent has a private valuation for connecting to the source (i.e., the maximum cost that it is willing to share). To maximize social welfare (i.e., the difference between the agents' valuations and the total connection cost), the agents need to report their valuations, but they may misreport for their own interest.

To minimize the total connection cost and maximize social welfare, we design cost sharing mechanisms on general networks that can prevent the two strategic behaviors. One difficulty lies in the conflict that the mechanism designer wants to use all the edges to minimize the total connection cost, but agents have the motivation to cut their adjacent edges to reduce their cost share. This essentially reflects the conflict between the system's optimality and the agents' self-interests. Another difficulty lies in the conflict that the mechanism designer wants to use truthful valuations to select agents with maximum social welfare, while agents have the motivation to misreport their valuations to reduce their cost share.

To combat the challenges, we first show that if we further require efficiency (the set of selected agents has the maximal social welfare) and budget balance (the sum of all agents' cost share equals the total cost), then it is impossible to prevent the above manipulations. However, we could achieve efficiency and budget balance separately.

Therefore, we propose two mechanisms to prevent new manipulations and to achieve either efficiency or budget balance. The first mechanism selects the agents based on their social welfare inspired by the Vickrey-Clarke-Groves (VCG) mechanism~\cite{vickrey1961counterspeculation,clarke1971multipart,groves1973incentives} and each agent pays the minimum reported valuation that enables it to be selected. The second selects the agents iteratively and the total connection cost in each iteration is shared equally among the agents selected in this iteration.
We also show that these mechanisms satisfy other desirable properties studied in the literature~\cite{bogomolnaia2010sharing,bergantinos2015characterization,gomez2017monotonic,norde2019degree,todo2020split}.

\section{Related Work}
\label{anay}
There is rich literature on the classic cost sharing problem, which did not consider private valuation and connection control. Some studies treated the problem from the perspective of a non-cooperative game. Berganti{\~n}os and Lorenzo~\cite{bergantinos2004non} studied the Nash equilibrium of the problem and further they~\cite{bergantinos2008noncooperative} studied the Nash equilibrium with budget restriction. Tijs and Driessen~\cite{tijs1986game} proposed the cost gap allocation (CGA) method, but it only applies to complete graphs. Bird~\cite{bird1976cost}, Dutta and Kar~\cite{dutta2004cost}, Norde \textit{et al.}~\cite{norde2004minimum}, Tijs \textit{et al.}~\cite{tijs2006obligation} and Hougaard \textit{et al.}~\cite{hougaard2010decentralized} provided cost sharing mechanisms based on the minimum spanning tree of a graph.

However, they do not satisfy truthfulness since the agents can change the minimum spanning tree by cutting their adjacent edges to reduce their cost share. We take the Bird rule~\cite{bird1976cost} for an example to show the problem. Under the Bird rule, the cost share of an agent is the cost of the edge that connects it to the (growing) spanning tree by Prim's algorithm~\cite{prim1957shortest} starting from the source. Consider Figure~\ref{Bird_rule}, when the agent $b$ does not cut the edge $(a,b)$, its cost share is 3. When it cuts the edge, its cost share is $2 < 3$. 
\begin{figure}[htb]
    \centering
    \includegraphics[width=5cm]{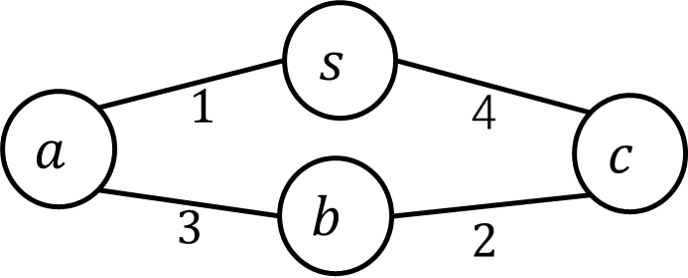}
    \caption{The $s$ represents the source, $a,b,c$ represent the agents, and the numbers on the edges represent the cost for the connectivity.}
    \label{Bird_rule}
\end{figure}

Other solutions treated the problem from the cooperative game perspective. They are all based on the Shapley value~\cite{shapley1953value} and differ in the definition of the value of each coalition.
We briefly describe them in the following.

Kar~\cite{kar2002axiomatization} proposed the Kar solution and defined the value of a coalition $S$ as the minimal cost of connecting all agents of $S$ to the source without going through the agents outside of $S$.
However, the Kar solution does not satisfy core selection.

To satisfy core selection, Berganti{\~n}os and Vidal-Puga~\cite{bergantinos2007fair} proposed the folk solution. They first compute the irreducible cost matrix, and then define the value of a coalition $S$ in the same way as the Kar solution. However, the folk solution throws away most information of the original graph. 

To obtain a core allocation without throwing away as much information as the folk solution, Trudeau~\cite{trudeau2012new} proposed the cycle-complete solution. They made a less extreme transformation of the cost matrix. 

Berganti{\~n}os and Vidal-Puga~\cite{bergantinos2007optimistic} proposed the optimistic game-based solution, where the value of a coalition $S$ is defined as the minimal cost of connecting all agents of $S$ to the source given that agents outside of $S$ are already connected to the source and the agents in $S$ can connect to the source through them. 

In summary, the existing solutions for the classic cost sharing problem on complete graphs do not consider the situation where agents need to report their valuations and adjacent edges, as is shown before, they do not guarantee to satisfy feasibility and truthfulness.

\section{The Model}
\label{The model}
We consider a cost sharing problem to connect the nodes in a weighted undirected graph $G=\langle V \cup \{s\},E \rangle$. The weight of the edge $(i,j) \in E$ denoted by $c_{(i,j)} \geq 0$ represents the cost to use the edge to connect $i$ and $j$. All the nodes in set $V$ want to connect to the source node $s$. The total cost of the connectivity has to be shared among all connected nodes except for $s$. Each node $i \in V$ has a private valuation $v_i \geq 0$, which is the maximum cost that it is willing to share. 

Given the graph, the minimum cost of connecting the nodes is the weight of the minimum Steiner tree~\cite{mehlhorn1988faster} (we assume that the graph is connected). The minimum Steiner tree of a set of nodes is a tree with the minimum weight that contains these nodes (it may include the nodes outside the set). The question here is how the nodes share this cost. We also consider two natural strategic behaviors of each node except for the source, i.e., cutting its adjacent edges and misreporting its valuation. An edge $(i,j)$ cannot be used for connectivity if $i$ or $j$ cuts it. 
Our goal is to design cost sharing mechanisms to incentivize nodes to report their valuations truthfully and also offer all their adjacent edges so that we can use all the edges to minimize the total cost of the connectivity. 

Formally, let $e_i (i \in V \cup \{s\})$ be the set of $i$'s adjacent edges and $\theta_i=(e_i,v_i)$ be the {\em type} of $i$. Let $\theta=(\theta_1,\cdots,\theta_{|V|+1})$ be the type profile of all nodes including the source $s$ (the valuation of $s$ is $null$). We also write $\theta=(\theta_i,\theta_{-i})$, where $\theta_{-i}=(\theta_1,\cdots,\theta_{i-1},\theta_{i+1},\cdots,\theta_{|V|+1})$ is the type profile of all nodes except for $i$. Let $\Theta_i$ be the type space of $i$ and $\Theta$ be the type profile space of all nodes (which generates all possible graphs containing $V \cup \{s\}$). 

We design a cost sharing mechanism that asks each node to report its valuation and the set of its adjacent edges that can be used for the connectivity. Let $\theta_i'=(e_i',v_i')$ be the report of $i$ where $e_i' \subseteq e_i$ and $v_i' \geq 0$, and $\theta'=(\theta_1',\cdots,\theta_{|V|+1}')$ be the report profile of all nodes. 
Given a report profile $\theta' \in \Theta$, the graph induced by $\theta'$ is denoted by $G(\theta')=\langle V \cup \{s\} , E(\theta') \rangle \subseteq \langle V \cup \{s\}, E \rangle$, where $E(\theta') = \{(i,j)| (i,j) \in (\theta_i' \cap \theta_j')\}$. Finally, let $r_i(\theta')\subseteq V$ be the set of $i$'s neighbour nodes.  

\begin{definition}
\label{def2}
A cost sharing mechanism consists of a node selection policy $g: \Theta \rightarrow 2^V$, an edge selection policy $f: \Theta \rightarrow 2^E$, and a cost sharing policy $x: \Theta \rightarrow \mathbb{R}^{|V|}$. Given a report profile $\theta' \in \Theta$, $g(\theta')\subseteq V$ selects the nodes to be connected, $f(\theta') \subseteq E(\theta')$ selects the edges to connect the selected nodes $g(\theta')$, and $x(\theta')=(x_i(\theta'))_{i \in V}$, where $x_i(\theta')$ is the cost share of $i$, which is zero if $i \notin g(\theta')$.
\end{definition}


For simplicity, we use $(g,f,x)$ to denote a cost sharing mechanism. Given a report profile $\theta' \in \Theta$, the utility of a node $i \in V$ under $(g,f,x)$ is defined as
\begin{equation*}
    u_i(\theta')=
    \begin{cases}
    v_i-x_i(\theta')& \text{if $i \in g(\theta')$,}\\
    0& \text{otherwise.}
    \end{cases}
\end{equation*}
In the following, we introduce the desirable properties of a cost sharing mechanism.

Feasibility requires that the cost share of each node is not over its reported valuation.

\begin{definition}
	\label{feasible}
	A cost sharing mechanism $(g,f,x)$ satisfies \textbf{feasibility} if $x_i(\theta') \leq v_i'$ for all $i \in V$, for all $\theta' \in \Theta$. 
\end{definition}

Truthfulness states that each node cannot increase its utility by cutting its adjacent edges and misreporting its valuation. Note that the source does not behave strategically in this setting. 

\begin{definition}
	\label{defic}
	A cost sharing mechanism $(g,f,x)$ satisfies \textbf{truthfulness} if $u_i((\theta_i, \theta'_{-i})) \geq u_i((\theta_i', \theta'_{-i})),$ for all $i\in V$, for all $\theta_i, \theta_i' \in \Theta_i$, and for all $\theta_{-i}'\in \Theta_{-i}=\Theta \setminus \Theta_i$. 
\end{definition}

Individual rationality requires that each node's utility is non-negative when it reports its type truthfully no matter what the others do. 
\begin{definition}
\label{defir}
A cost sharing mechanism $(g,f,x)$ satisfies \textbf{individual rationality (IR)} if $u_i(\theta_i,\theta_{-i}') \geq 0$ for all $i \in V$, for all $\theta_i \in \Theta_i$, and for all $\theta_{-i}' \in \Theta_{-i}=\Theta \setminus \Theta_i$.  
\end{definition}

Utility monotonicity states that for each selected node, its utility will weakly decrease if the cost of one of its adjacent edges increases under the same report profile.

\begin{definition}

A cost sharing mechanism $(g,f,x)$ satisfies \textbf{utility monotonicity (UM)} if $u_i(\theta') \geq u_i^{+}(\theta')$ for all $\theta' \in \Theta$ and for all $i \in V$, where $u_i^{+}(\theta')$ is $i$'s utility when the cost of the edge $(i,j) \in \theta_i'$ increases.
\end{definition}

We also require that the sum of all nodes' cost share equals the total cost of the selected edges for any report profile. That is, the mechanism has no profit or loss.

\begin{definition}
A cost sharing mechanism $(g,f,x)$ satisfies \textbf{budget balance (BB)} if $\sum_{i \in V}x_i(\theta')=\sum_{(i,j) \in f(\theta')}c_{(i,j)}$ for all $\theta' \in \Theta$.
\end{definition}

The ranking property requires that for any nodes $i$ and $j$ that have the same reported valuations and the same neighbour nodes except for $i$ and $j$, if the cost of the edge $(i,k)$ is less expensive than the edge $(j,k)$ for any neighbour node $k$, then the utility of $i$ should be larger than $j$.

\begin{definition}
A cost sharing mechanism $(g,f,x)$ satisfies \textbf{ranking} if for all $\theta' \in \Theta$, for all $i, j \in V$ with $r_i(\theta') \setminus \{j\} = r_j(\theta') \setminus \{i\} $ and $v_i'=v_j'$ (assume $v_i'=v_i$ and $v_j'=v_j$), we have $c_{(i,k)} \leq c_{(j,k)}$ for all $k \in r_i(\theta') \setminus \{j\}$ implies $u_i(\theta') \geq u_j(\theta')$.
\end{definition}


Symmetry says nodes that play the same role obtain the same utility.
\begin{definition}
A cost sharing mechanism $(g,f,x)$ satisfies \textbf{symmetry} if for all $\theta' \in \Theta$, for all $i, j \in V$ with $r_i(\theta') \setminus \{j\} = r_j(\theta') \setminus \{i\} $ and $v_i'=v_j'$ (assume $v_i'=v_i$ and $v_j'=v_j$), we have $c_{(i,k)} = c_{(j,k)}$ for all $k \in r_i(\theta') \setminus \{j\}$ implies $u_i(\theta') = u_j(\theta')$.
\end{definition}

Finally, each node's cost share should be non-negative.

\begin{definition}
A cost sharing mechanism $(g,f,x)$ satisfies \textbf{positiveness} if $x_i(\theta')\geq 0$ for all $i \in V$ and for all $\theta' \in \Theta$.  
\end{definition}

In the rest of the paper, we design cost sharing mechanisms to satisfy the above properties.


\section{Impossibility Results}
In this section, we establish some impossibility results. We first introduce some extra notions. 

\begin{definition}
For a given subset $S \subseteq V$, the social welfare (SW) of $S$ is $$SW(S)=\sum_{i \in S}v_i'-C(S),$$ where $v_i'$ is the reported valuation of node $i$ and $C(S)$ is the minimum cost of connecting all the nodes in $S$ (i.e., the weight of the minimum Steiner tree of $S\cup\{s\}$).
\end{definition}

Intuitively, social welfare represents the profit of the selected nodes. 

\begin{definition}
Given $\theta' \in \Theta$, a mechanism $(g,f,x)$ satisfies \textbf{efficiency} if it selects $g(\theta') \subseteq V$ such that its social welfare is maximized, i.e., $$SW(g(\theta'))={\underset{\forall S \subseteq V}{{\max}\ SW(S)}.}$$ 
\end{definition}
For simplicity, we use $\delta(V)$ to denote a subset of $V$ that has the maximal social welfare, i.e.,
$$
\delta(V)={\underset{ \forall S \subseteq V}{{\arg\max}\ SW(S)}.}
$$
The computation of $\delta(V)$ is described as follows.
        
\begin{framed}
\label{maxSW}
 \noindent\textbf{Algorithm 1: Compute $\delta(V)$}
 
 \noindent\rule{\textwidth}{0.5pt}
 
 \noindent\textbf{Input}: A report profile $\theta'$ and a set of nodes V
 
 \noindent\rule{\textwidth}{0.5pt}
 
 \begin{enumerate}
 \item Set: $\delta(\emptyset)=\emptyset$, $SW(\emptyset)=0$, and \\ \textcolor{white}{Set:} $P(V)$ to be the power set of $V$.
 \item Sort all the elements of $P(V)$ in an ascending order by their cardinalities.
 \item For $S \in P(V)\setminus \{\emptyset\}$:
 \begin{itemize}
     \item Get $Q(S)=\{S'|S'\subset S, |S'|=|S|-1 \}$.
     \item Let $\delta(S) =\mathop{\arg\max}\limits_{S'\in Q(S)} SW(\delta(S'))$.
     \item \textbf{If} $\sum_{i\in S} v_i'-C(S)\geq SW(\delta(S))$, set $\delta(S)=S$.
     \item Get $SW(\delta(S))=\sum_{i\in \delta(S)} v_i'-C(\delta(S))$.
 \end{itemize}
 \end{enumerate}
 
 \noindent\rule{\textwidth}{0.5pt}
 
 \noindent\textbf{Output}: The subset $\delta(V)$, \\ \textcolor{white}{\noindent\textbf{Output:}:} the maximum social welfare $SW(\delta(V))$
\end{framed}

\begin{figure}[htb]
    \centering
    \includegraphics[width=5cm]{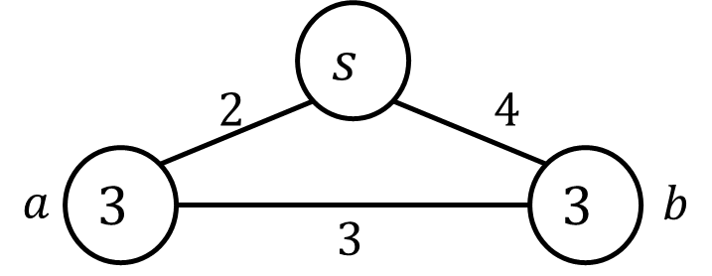}
    \caption{The $s$ represents the source, $a$ and $b$ represent the nodes, the numbers in the circles represent the reported valuations of nodes, and the numbers on the edges represent the cost for the connectivity.}
	 \label{deltaS}
\end{figure}
A running example of Algorithm 1 is given in Figure~\ref{deltaS}. Assume $c_{(s,a)}=2, c_{(s,b)}=4, c_{(a,b)}=3, v_a'=3, v_b'=3$. By Algorithm 1, we have $\delta(\{a\})=\{a\}, \delta(\{b\})=\emptyset$ and $\delta(\{a,b\})=\{a,b\}$.

\begin{proposition}
\label{impo2}
There exists no cost sharing mechanism which satisfies truthfulness, feasibility, efficiency, and budget balance simultaneously.
\end{proposition}
\begin{figure}[htb]
    \centering
    \includegraphics[width=5cm]{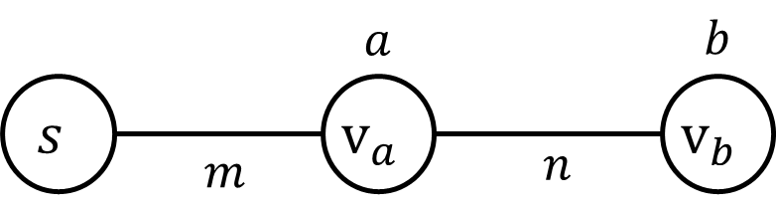}
    \caption{The $s$ represents the source and $a,b$ are the nodes with the valuations $v_a,v_b$ respectively. The cost of the edges $(s,a)$ and $(a,b)$ are $m$ and $n$ respectively.}
	\label{impossible_ex12}
\end{figure}
\begin{proof}

We only need to consider a simple line graph in Figure~\ref{impossible_ex12}. We show that when feasibility, efficiency, and budget balance are satisfied, truthfulness will be violated. Assume that $v_a>m$ and $v_b > m+n$. 
By efficiency, when $a$ and $b$ truthfully report $ v_a'=v_a, v_b'=v_b$, $a$ and $b$ are both selected by the mechanism since $(v_a + v_b) - (m+n)> v_a -m$.
\begin{itemize}
    \item When $x_a(\theta')> 0$, if node $a$ reports $v_a''=0$, by feasibility, we have $x_a(\theta'')=0$ and the utility of $a$ increases. Hence, node $a$ has the motivation to misreport.
    \item When $x_a(\theta')=0$, by budget balance, we have $x_b(\theta')=m+n$. If node $b$ reports $v_b''=n$, by feasibility and budget balance, we have $x_b(\theta'')=n$, and the utility of $b$ increases. Hence, node $b$ has the motivation to misreport.         
\end{itemize}

Therefore, node $a$ or node $b$ has the motivation to misreport its valuation, i.e., truthfulness is violated.
\end{proof}

We further show that when truthfulness, feasibility, and budget balance are satisfied, the maximal social welfare cannot be approximated.
\begin{definition}
\label{lower}
A mechanism is $\alpha^{lb}$-approximate ($\alpha^{lb} \in (0,1)$) to the social welfare if $SW(g(\theta'))\geq \alpha^{lb} \cdot SW^*(S)$, where $SW^*(S)$ is the maximal social welfare given $\forall \theta' \in \Theta$ and $lb$ represents that $\alpha^{lb}$ is a lower bound of the ratio $\frac{SW(g(\theta'))}{SW^*(S)}$.
\end{definition}

\begin{definition}
\label{upper}
A mechanism is $\beta^{ub}$-approximate ($\beta^{ub} \in (0,1)$) to the social welfare if $SW(g(\theta')) \leq \beta^{ub} \cdot SW^*(S)$, where $SW^*(S)$ is the maximal social welfare given $\forall \theta' \in \Theta$ and $ub$ represents that $\beta^{ub}$ is an upper bound of the ratio $\frac{SW(g(\theta'))}{SW^*(S)}$.
\end{definition}

\begin{proposition}
There exists no cost sharing mechanism that satisfies truthfulness, budget balance, feasibility, and $\alpha^{lb}$-approximation ($\beta^{ub}$-approximation) simultaneously.
\end{proposition}

\begin{proof}

It suffices to consider a simple line graph in Figure~\ref{impossible_ex12}. Without loss of generality, assume that $v_a > m,v_b=n+p (p>0),c_{(s,a)}=m$ and $c_{(a,b)}=n$. When $a$ and $b$ are both selected by the mechanism, the maximum social welfare is ($v_a - m + p$).

Next, we show when truthfulness, feasibility, and budget balance are satisfied, $\alpha^{lb}$-approximation and $\beta^{ub}$-approximation will be violated. By Proposition~\ref{impo2}, any mechanism cannot select both $a$ and $b$ and it can only select $a$. Then the social welfare is $(v_a - m)$. Hence, the ratio equals $\frac{v_a-m}{v_a-m+p}$. Letting $p \rightarrow \infty$, then the ratio approaches 0. Therefore, the required $\alpha^{lb}$ does not exist. Again letting $p \rightarrow 0$, then the ratio approaches $1$, which means that the required $\beta^{ub}$ does not exist.

Therefore, any cost sharing mechanism cannot satisfy truthfulness, budget balance, feasibility, and $\alpha^{lb}$-approximation ($\beta^{ub}$-approximation) simultaneously.
\end{proof}

We further consider the deficit of any mechanism that satisfies truthfulness, feasibility, and efficiency. We introduce a concept called budget balance ratio to evaluate it.

\begin{definition}
\label{bbr}
A mechanism has a budget balance ratio (BBR) called $\gamma \in (0,1]$ if $\sum_{i\in g(\theta')}x_i(\theta')\geq \gamma \cdot C(g(\theta'))$, $\forall \theta' \in \Theta$. 
\end{definition}

\begin{proposition}
\label{BBR}
A cost sharing mechanism that satisfies truthfulness, feasibility and efficiency does not have a BBR. 
\end{proposition}
\begin{figure}[htb]
    \centering
    \includegraphics[width=5cm]{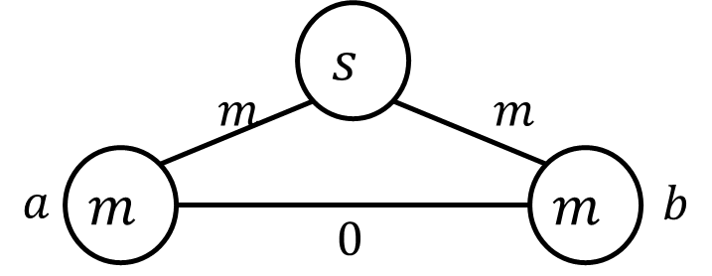}
    \caption{The $s$ represents the source and $a,b$ are the nodes with the valuation $m$. The cost of the edges $(s,a)$, $(a,b)$ and $(s,b)$ are $m,0,m$ respectively.}
	 \label{Impossible_ex3}
\end{figure}
\begin{proof}
According to Definition~\ref{bbr}, a cost sharing mechanism having a BBR $\gamma \in (0,1]$ needs to satisfy the following: $\forall \theta' \in \Theta$, $\frac{\sum_{i\in g(\theta')}x_i(\theta')}{C(g(\theta'))}\geq \gamma > 0$. So, to prove the proposition, it suffices to find a $\theta'$ such that $\frac{\sum_{i\in g(\theta')}x_i(\theta')}{C(g(\theta'))}= 0$.

Without loss of generality, as shown in Figure~\ref{Impossible_ex3}, we assume $V=\{a,b\}$, $v_a=v_b=m$, $c_{(s,a)}=c_{(s,b)}=m$, $c_{(a,b)}=0$. By efficiency, the mechanism should select both $a$ and $b$. By truthfulness, each node offers all its adjacent edges and reports its valuation truthfully. By feasibility and truthfulness, we have $x_a(\theta')=x_b(\theta')=0$. Thus, we have $\frac{\sum_{i\in g(\theta')}x_i(\theta')}{C(g(\theta'))}=\frac{0}{m}=0$.
\end{proof}

Note that there exists a trivial cost sharing mechanism where each node pays 0 and all the nodes in $V$ are selected by the mechanism. This mechanism satisfies truthfulness and feasibility but does not satisfy efficiency and budget balance. 

By Proposition~\ref{impo2}, a cost sharing mechanism cannot simultaneously satisfy truthfulness, feasibility, efficiency, and budget balance. Therefore, we propose two feasible mechanisms satisfying truthfulness, respectively together with efficiency and budget balance in the following sections. 

We summarize the impossibility results and our mechanisms in Table~\ref{tab0}.

\begin{table}[htb]
\caption{"Null" means that there exists no mechanism satisfying all the marked properties in the column. Our mechanisms RSM and CVM satisfy the marked properties in the column.}
\label{tab0}
\resizebox{0.9\linewidth}{!}{
\begin{tabular}{|c|c|c|c|c|}
\hline
Truthfulness & Feasibility & Budget Balance & Efficiency & Mechanism\\
\hline
\checkmark&\checkmark & \checkmark & \checkmark & NULL\\
\checkmark&\checkmark & \checkmark & $\alpha^{lb}$($\beta^{ub}$) & NULL\\
\checkmark&\checkmark & $\gamma$ & \checkmark & NULL\\
\checkmark&\checkmark & & \checkmark& CVM\\
\checkmark&\checkmark &           \checkmark & & RSM\\
\hline
\end{tabular}}

\end{table}


\section{Critical Value Based Mechanism}
\label{csm-ub}
In this section, we propose a cost sharing mechanism that satisfies truthfulness, feasibility, and efficiency but does not satisfy budget balance. In addition, we show that it also satisfies other desirable properties. 

The key ideas of the mechanism are as follows. First, find out the node set which has maximal social welfare. Second, for each node in the set, compute its critical value (CV), i.e., the minimal reported valuation that keeps it in the set. 

Finally, let the cost share of the node in the set equal its critical value and the others' cost share is 0. 

The computation of the minimum reported valuation of each node $i \in g(\theta')$ is as follows. We first compute $\delta(g(\theta')\setminus\{i\})$, the set of nodes that maximizes the social welfare when node $i$ is not considered. Then we compute the social welfare of $g(\theta')$ and $\delta(g(\theta')\setminus\{i\})$. Next, we find out the minimum reported valuation of $i$ that keeps it in $g(\theta')$ and guarantees $SW(g(\theta'))=SW(\delta(g(\theta')\setminus\{i\}))$.  

The mechanism is formally described as follows. A running example is given after the algorithm.  

\begin{framed}
 \noindent\textbf{Critical Value Based Mechanism (CVM)}
 
 \noindent\rule{\textwidth}{0.5pt}
 
 \noindent\textbf{Input}: A report profile $\theta'$ and a graph $G(\theta')$
 
 \noindent\rule{\textwidth}{0.5pt}
 
 \begin{enumerate}
 \item Run Algorithm 1 and get $g(\theta') = \delta(V)$.
 \item Compute the minimum Steiner tree of $g(\theta') \cup \{s\}$ and \\ set $f(\theta')$ to be the set of edges in the tree.
 \item For $i \in g(\theta')$:
 \begin{itemize}
     \item Compute node $i$'s critical value
		\begin{equation}
		\label{critical}
		\begin{split}
	    CV_i(\theta')=&(\sum_{j \in \delta(g(\theta')\setminus \{i\})}v_j'-C(\delta(g(\theta')\setminus \{i\})))\\
		 &-(\sum_{k \in g(\theta')\setminus \{i\}}v_k'-C(g(\theta'))),\\    
		\end{split}
		\end{equation}
		where $\delta(\cdot)$ is defined in Algorithm 1.
 \end{itemize}
 	\item Set $x_i(\theta')=CV_i(\theta')$.
 \end{enumerate}
 
 \noindent\rule{\textwidth}{0.5pt}
 
 \noindent\textbf{Output}: The selected nodes $g(\theta')$, the selected edges $f(\theta')$, the cost sharing $x(\theta')$
\end{framed}
\begin{figure}[htb]
    \centering
\includegraphics[width=9cm]{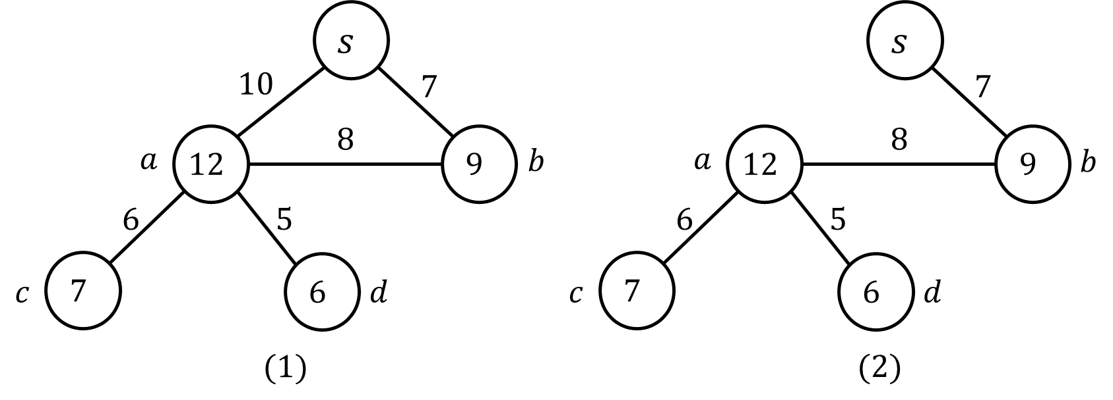}
    \caption{The left figure is $G(\theta')$ and the right figure is the minimum Steiner tree of $g(\theta') \cup \{s\}$. \iffalse The $s$ represents the source, $a,b,c,d$ represent the nodes, the numbers on the edges represent the cost, and the numbers in the circles represent the valuations of the nodes.\fi}
	 \label{1110}
\end{figure}

\begin{example}
The graph $G(\theta')$ generated by a report profile $\theta' \in \Theta$ is shown in Figure~\ref{1110}(1). First, run Algorithm 1 and obtain $\delta(S)$ for all $S\subseteq V$. Especially, we have $g(\theta')=\delta(V)=\{a,b,c,d\}$ and $f(\theta')=\{(s,b),(a,b),(a,c),(a,d)\}$. Then we compute each node's cost share. Taking the node $a$ for an example, we have $\delta(g(\theta')\setminus\{a\})=\{b\}$ and by Equation~(\ref{critical}), $x_a(\theta')=v_b'-c_{(s,b)}-(v_b'+v_c'+v_d'-c_{(s,b)}-c_{(a,b)}-c_{(a,c)}-c_{(a,d)})=9-7-(9+6+7-7-8-6-5)=6$. Similarly, we have $x_b(\theta')=5, x_c(\theta')=6$ and $x_d(\theta')=5$. Thus we have $x(\theta')=(6,5,6,5)$.
\end{example}

\subsection{Properties of CVM}
\label{pro}
Now we show some nice properties of the critical value based mechanism. 

\begin{theorem}
	\label{thethe3}
	The critical value based mechanism satisfies truthfulness. 
\end{theorem}
\begin{proof}
First, we prove that each node $i \in V$ will report its valuation truthfully (i.e., $v_i'=v_i$). When $i$ truthfully reports its valuation, there are two cases.
\begin{itemize}
    \item $i \notin g(\theta')$. We have $u_i(\theta')=0$. If node $i$ reports $v_i'< v_i$, it is still not selected and the utility does not change. Otherwise ($v_i' > v_i$), there are two possibilities.
    \begin{itemize}
        \item It is still not selected and the utility does not change.
        \item It is selected. By the proposed mechanism, since $i$'s critical value is larger than $v_i$, its utility is negative. Thus the utility decreases.
    \end{itemize}
    \item $i \in g(\theta')$. We have $u_i(\theta') \geq 0$. If node $i$ reports $v_i'>v_i$, it is still selected and the utility does not change. Otherwise ($v_i' \leq v_i$), there are two possibilities.
    \begin{itemize}
        \item It is still selected. The utility does not change.
        \item It is not selected and its utility is 0. So the utility decreases.
    \end{itemize}
\end{itemize}

Second, we prove that each node $i \in V$ will report its adjacent edges truthfully (i.e., $e_i'=e_i$). When $i$ truthfully reports its adjacent edges, there are two cases. 
\begin{itemize}
    \item $i \notin g(\theta')$. We have $u_i(\theta')=0$. If node $i$ reports $e_i' \neq e_i$, its cost share will weakly increase and thus it is still not selected. Hence, the utility does not change.
    \item $i \in g(\theta')$. We have $u_i(\theta')\geq 0$. If node $i$ reports $e_i' \neq e_i$, there are two possibilities. 
    \begin{itemize}
        \item It is not selected. Obviously, its utility decreases.
        \item It is still selected. For simplicity, let $S_1=g(\theta'), S_2=g(\theta''), S_3=S_1\setminus \delta(S_1\setminus\{i\}), S_4=S_2\setminus \delta(S_2\setminus\{i\})$ where $\theta''=((e_i',v_i'),\theta'_{-i})$ and $\theta'=((e_i,v_i'),\theta'_{-i})$. Then by Equation~(\ref{critical}) we have 
        \begin{align*}
        CV_i(\theta') &= C(S_1)-C(\delta(S_1\setminus\{i\}))-\sum_{j\in S_3\setminus\{i\}}v_j', \\
        CV_i(\theta'') &=C'(S_2)-C'(\delta(S_2\setminus\{i\}))-\sum_{j\in S_4\setminus\{i\}}v_j',
        \end{align*}
        where $C'(\cdot)$ denotes the value function when $i$ misreports $e_i'$. Since $S_1$ maximizes the social welfare under $\theta'$, we have
        \begin{equation*}
            \begin{split}
        \sum_{j\in S_3}v_j' -(C(S_1)-C(\delta(S_1\setminus\{i\}))) &\\
        \geq \sum_{j\in S_4}v_j' -(C(S_2)-C(\delta(S_2\setminus\{i\}))).
        \end{split}
        \end{equation*}
        For the set $S_4$, the increment of SW will decrease since the set of available edges of $S_4$ is reduced. Then we have
        \begin{equation*}
            \begin{split}
        \sum_{j\in S_4}v_j' -(C(S_2)-C(\delta(S_2\setminus\{i\}))) 
        &\\\geq 
        \sum_{j\in S_4}v_j' -(C'(S_2)-C'(\delta(S_2\setminus\{i\}))).
            \end{split}
        \end{equation*}
        Therefore, 
        \begin{equation*}
            \begin{split}
                 \sum_{j\in S_3}v_j' -(C(S_1)-C(\delta(S_1\setminus\{i\})))
                 &\\
                 \geq \sum_{j\in S_4}v_j' -(C'(S_2)-C'(\delta(S_2\setminus\{i\}))).
            \end{split}
        \end{equation*}
        Eliminating $v_i'$, we have $CV_i(\theta') \leq CV_i(\theta'')$.
    \end{itemize}
\end{itemize}
So the cost share of $i$ weakly increases when misreporting its adjacent edges. Therefore, $i$'s utility weakly decreases when misreporting its adjacent edges.
\end{proof}

\begin{theorem}
\label{effi}
The critical value based mechanism satisfies feasibility. 
\end{theorem}
\begin{proof}
According to the mechanism and Algorithm 1, the participation of each selected node can increase social welfare. Because the critical value is the minimum reported valuation that keeps the node being selected, its cost share is less than or equal to its reported valuation. For the other nodes, their cost share is 0, which is less than their reported valuation.      
\end{proof}

\begin{theorem}
\label{}
The critical value based mechanism satisfies individual rationality.
\end{theorem}
\begin{proof}
Given a report profile $\theta' \in \Theta$, for each node $i \in V\setminus g(\theta')$, we have $u_i(\theta')=0$. For each node $i \in g(\theta')$, by Theorem~\ref{effi}, we have $x_i(\theta')\leq v_i'$. According to Theorem~\ref{thethe3}, $v_i'=v_i$. Hence, $u_i(\theta')=v_i-x_i(\theta')\geq 0$.
\end{proof}

\begin{theorem}
\label{core}
The critical value based mechanism satisfies efficiency.
\end{theorem}
\begin{proof}
According to the mechanism, it is obvious that the set of selected nodes can maximize the social welfare.
\end{proof}

\begin{theorem}
\label{sym}
The critical value based mechanism satisfies positiveness.  
\end{theorem}
\begin{proof}
We prove the statement by contradiction. According to the mechanism and Equation~(\ref{critical}), we have
\begin{equation*}
		\begin{split}
		x_i(\theta')
		    =CV_i(\theta')    
		  = (C(g(\theta'))-C(\delta(g(\theta')\setminus \{i\})))
		 -\Delta.\\
		\end{split}
		\end{equation*}

where $$\Delta =\sum_{k \in g(\theta')\setminus \{i\}}v_k'-\sum_{j \in \delta(g(\theta')\setminus \{i\})}v_j' $$

If $x_i(\theta') \leq 0$, then we have $$C(g(\theta'))-C(\delta(g(\theta')\setminus \{i\}))\leq \Delta$$
Since $$C(g(\theta') \setminus \{i
\}) \leq C(g(\theta')),$$ we have $$C(g(\theta') \setminus \{i
\})-C(\delta(g(\theta')\setminus \{i\}))\leq \Delta$$ Therefore, the nodes in $(g(\theta')\setminus \{i\}-\delta(g(\theta')\setminus \{i\}))$ can be selected by the mechanism. By the definition of $\delta(\cdot)$, they cannot be selected by the mechanism. This leads to a contradiction.  
\end{proof}

\begin{theorem}
\label{symmetry}
The critical value based mechanism satisfies symmetry.
\end{theorem}

\begin{proof}
We need to show that, given $\theta' \in \Theta$ and $i,j \in V$ with $r_i(\theta') \setminus \{j\} = r_j(\theta') \setminus \{i\} $ and $v_i=v_j$, $c_{(i,k)} = c_{(j,k)}$ ($\forall k \in r_i(\theta') \setminus \{j\}$) implies $u_i(\theta') = u_j(\theta')$. If $i,j\notin g(\theta')$, we have $u_i(\theta')=u_j(\theta')=0$. If $i,j\in g(\theta')$, by the condition of symmetry, we have $x_i(\theta')=x_j(\theta')$. Since $v_i=v_j$, we have $u_i(\theta')=u_j(\theta')$.
\end{proof}

\begin{theorem}
\label{cm}
The critical value based mechanism satisfies utility monotonicity.
\end{theorem}
\begin{proof}
For any node $i \in V$, given $\theta' \in \Theta$, $j \in V$ such that $(i,j) \in E$, we use $g^+(\theta')$ to denote the set of selected nodes when $c_{(i,j)}$ increases.

If $i \notin g(\theta')$, then $u_i(\theta')=0$, and $i \notin g^+(\theta')$ according to CVM. So its utility does not change.

If $i\in g(\theta')$, then $u_i(\theta')\geq 0$. There are two cases.
    \begin{itemize}
        \item $i\notin g^+(\theta')$. Then its utility weakly decreases.
        \item $i\in g^+(\theta')$. Its cost share becomes $$(\sum_{j \in \delta(g^+(\theta')\setminus \{i\})}v_j-C(\delta(g^+(\theta')\setminus \{i\})))-(\sum_{k \in g^+(\theta')\setminus \{i\}}v_k-C(g^+(\theta'))).$$ By the similar analysis to the second part in the proof of Theorem~\ref{thethe3}, we know the utility of $i$ weakly decreases.
    \end{itemize}
\end{proof}

\begin{theorem}
The critical value based mechanism satisfies ranking.
\end{theorem}

\begin{proof}
We need to show that, given $\theta' \in \Theta$ and $i,j \in V$ with $r_i(\theta') \setminus \{j\} = r_j(\theta') \setminus \{i\} $ and $v_i=v_j$, $c_{(i,k)} \leq c_{(j,k)}$ ($\forall k \in r_i(\theta') \setminus \{j\}$) implies $u_i(\theta') \geq u_j(\theta')$. For nodes $i$ and $j$, there are three cases.
\begin{itemize}
    \item $i,j\notin g(\theta')$. We have $u_i(\theta')=u_j(\theta')=0$.
    \item $i\in g(\theta')$ but $j \notin g(\theta')$. By individual rationality, we have $u_i(\theta')\geq 0=u_j(\theta')$.
    \item $i,j\in g(\theta')$. Since $v_i=v_j$, it suffices to prove $x_i(\theta') \leq x_j(\theta')$. From Equation~(\ref{critical}), we know the last two terms of the expressions of $x_i(\theta')$ and $x_j(\theta')$ are equal. Next, we compare $\sum_{j \in \delta(g(\theta') \setminus \{i\})}v_j-C(\delta(g(\theta')\setminus \{i\}))$ with $\sum_{i \in \delta(g(\theta')\setminus \{j\})}v_i-C(\delta(g(\theta') \setminus \{j\}))$. The former represents the maximum social welfare of $g(\theta') \setminus \{i\}$ and the latter represents the maximum social welfare of $g(\theta') \setminus \{j\}$. By the symmetry of $i$ and $j$ in the graph, the condition of ranking, and the proof of Theorem~\ref{cm}, the latter is larger than or equal to the former. Therefore, we have $x_i(\theta') \leq x_j(\theta')$.   
\end{itemize}
\end{proof}

\section{Repeated Selection Mechanism}
\label{csm-lb}
The CVM defined in Section~\ref{csm-ub} satisfies truthfulness, feasibility, and efficiency but does not satisfy budget balance. In this section, we propose another cost sharing mechanism that satisfies truthfulness, feasibility, and budget balance but does not satisfy efficiency. Moreover, we show that it also satisfies other desirable properties.

We use the method of iterative optimization. In the first round (stage) of optimization, we find a subset of nodes and the minimum satisfying the constraints of feasibility and budget balance as the cost share of these nodes. In the following rounds of optimization, we consider the remaining nodes and add an extra constraint that the optimizing variable is larger than or equal to the cost share of the last round, which guarantees the truthfulness of the mechanism. The iterative process is repeated until all the nodes are considered.   



The proposed mechanism is called Repeated Selection Mechanism (RSM) formally described in the following.

\begin{framed}
\label{budget}
 \noindent\textbf{Repeated Selection Mechanism (RSM)}
 
 \noindent\rule{\textwidth}{0.5pt}
 
 \noindent\textbf{Input}: A report profile $\theta'$ and a graph $G(\theta')$
 
 \noindent\rule{\textwidth}{0.5pt}
 
 \begin{enumerate}
    \item For stage $t$ ($t = 0,1,2,\dots$), we introduce:
        \begin{itemize}
            \item $S^t$ : the set of nodes selected in stage $t$, 
            \item $M^t$ : the union of $S^0,S^1,\cdots,S^t$,
            \item $X^t$ : the cost share of every node in $S^t$,
            \item $W^t$ : the set of nodes which will not be considered after stage $t$, 
            \item $N^t$ : the set of remaining nodes after stage $t$, and
            \item $\mathcal{E}^t$: the set of edges of the minimum Steiner tree of $S^t$.
        \end{itemize}
    $\mathcal{E}$ is the set of selected edges during the process, i.e., the union of $\mathcal{E}^0,\mathcal{E}^1,\cdots,\mathcal{E}^t$.
    
    \item Stage 0: 
    
    Set
    $X^0=0$, $N^0=V$, $M^0=\emptyset$, $W^0=\emptyset$, and $\mathcal{E}=\emptyset$.
    
    \item For $t = 1, 2, \dots :$
    
    $\quad$Stage $t$:
    \begin{align*}
     Solve \quad\quad \mathop{min}\limits_{ S^t\subseteq N^{t-1} }&X^t \\
    s.t. \quad  &X^t \geq X^{t-1} \\ 
    &X^t\cdot|S^t| = C(S^t) \\
    &v_i'\geq X^t ,  \forall i\in S^t
    \end{align*}
    \begin{itemize}
        \item If there is a solution, set $x_i(\theta')=X^t$ ($\forall i\in S^t$).\\
        Then, update: $$W^t=\{i|v_i'<X^t\},$$ $$N^t=N^{t-1}\setminus \{S^t\cup W^t\},$$ $$M^t=M^{t-1}\cup S^t.$$
      Set:$$c_{(i,j)} = 0, \forall i,j \in M^t\cup \{s\}, i\neq j,$$
      $$\mathcal{E} =\mathcal{E}\cup \mathcal{E}^t.$$
      Note that "$c_{(i,j)}=0$" means that the remaining nodes can use the edge $(i,j)$ without paying cost.
    \item Else, set $g(\theta')=M^{t-1}$ and $f(\theta')=\mathcal{E}$. Break the loop.
    \end{itemize} 
 \end{enumerate}
 
 \noindent\rule{\textwidth}{0.5pt}
 
 \noindent\textbf{Output}: The selected nodes $g(\theta')$, the selected edges $f(\theta')$, the cost sharing $x(\theta')$ 
\end{framed}

A running example is given as follows.

\begin{example}
\label{ex4.1}
The graph $G(\theta')$ generated by a report profile $\theta' \in \Theta$ is shown in Figure~\ref{RSM_ex}. For stage 1, we have $X^1=3$ and $S^1=\{b\}$. Hence, $\mathcal{E}=\mathcal{E}^1=\{(s,b)\}$. For stage 2, we have $X^2=4$ and $S^2=\{a\}$. Hence, $\mathcal{E}^2=\{(a,b)\}$ and $\mathcal{E}=\mathcal{E}^1 \cup \mathcal{E}^2=\{(s,b),(a,b)\}$. For stage 3, we have $X^3=5$ and $S^3=\{c,d,e\}$. Hence, $\mathcal{E}^3=\{(a,d),(a,c),(b,e)\}$ and $\mathcal{E}=\mathcal{E}^1 \cup \mathcal{E}^2 \cup \mathcal{E}^3=\{(s,b),(a,b),(a,d),(c,d),(b,e)\}$. Since there does not exist $X^4$ satisfying the constraints, the proposed algorithm ends and we have $g(\theta')=\{a,b,c,d,e\},x(\theta')=(3,4,5,5,5)$ and $f(\theta')=\{(s,b),(a,b),(a,d),(c,d),(b,e)\}$.
\end{example}
\begin{figure}[htbp]
    \centering
    \includegraphics[width=9cm]{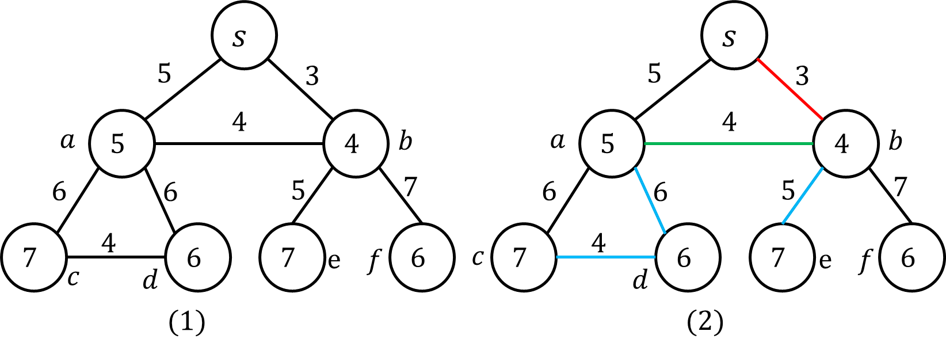}
     \caption{The left figure is $G(\theta')$. The red line in the right figure denotes the selected edge in the first stage, the green line denotes the selected edge in the second stage and the blue lines denote the selected edges in the third stage. \iffalse The $s$ represents the source and $a, b, c, d, e, f$ represent the nodes. The numbers on the edges represent the cost for the connectivity and the numbers in the circles represent the valuations.\fi}
	 \label{RSM_ex}
\end{figure}




\subsection{Properties of RSM}
We show the properties of RSM in this section.
\begin{theorem}
	\label{IC_2}
	The repeated selection mechanism satisfies truthfulness. 
\end{theorem}
\begin{proof} 
Firstly, we prove that each node $i \in V$ will report its adjacent edges truthfully. Denote two report profile by $\theta'=((e_i',v_i'),\theta_{-i}')$ where $e_i'=e_i$ and $\theta''=((e_i'',v_i'),\theta_{-i}')$. When $i$ truthfully reports its adjacent edges, there are two cases.
\begin{enumerate}
    \item $i \notin g(\theta')$. Then we have $u_i(\theta')=0$. If $i$ reports $e_i'' \subset e_i$, for any $S(i\in S)$, $C(S)$ will increase since the set of available edges shrinks. Hence, $i \notin g(\theta'')$ and the utility does not change.  
    \item $i \in g(\theta')$. Then $u_i(\theta')=v_i-x_i(\theta')\geq 0$. Assume that $i\in S^t$. If $i$ reports $e_i''\subset e_i$, we first prove the set of selected nodes before stage $t$ does not change due to $i$'s misreporting, i.e., $M^{t-1}=\hat{M}^{t-1}$, where $\hat{M}^{t-1}$ is the set of selected nodes before stage $t$ given $\theta''$. If node $i$ belongs to the minimum Steiner tree of $S^r$ in the stage $r$ ($r<t$), then it should have been selected in stage $r$, which leads to a contradiction to $i\in S^t$. Therefore, we know node $i$ does not belong to the minimum Steiner tree of $S^r$ for any stage $r$ ($r<t$). Then the selected edges of $M^{t-1}$ and $\hat{M}^{t-1}$ are the same, i.e., $M^{t-1}=\hat{M}^{t-1}$.
 
    Then we prove the utility of node $i$ weakly decreases due to $i$'s misreporting. Based on the above analysis, there are two cases for node $i$.
	\begin{itemize}
		       \item Node $i\in g(\theta'')$. Let $C'(S)$ denote the minimum cost of any set $S$ under $\theta''$. Since $C'(S)\geq C(S)$, we have $\frac{C'(S)}{|S|}\geq \frac{C(S)}{|S|}$. Hence, $x_i(\theta'')\geq x_i(\theta')$ and $u_i(\theta'')\leq u_i(\theta')$.
                \item Node $i\notin g(\theta'')$. Then we have $u_i(\theta'')=0\leq u_i(\theta')$. 
\end{itemize}
\end{enumerate}

	Secondly, we prove that each node $i \in V$ will report its valuation truthfully. Denote two report profile by $\theta'=((e_i',v_i'),\theta_{-i}')$ where $v_i'=v_i$ and $\theta''=((e_i',v_i''),\theta_{-i}')$. When $i$ truthfully reports its valuation, there are two cases.
	\begin{enumerate}
	    \item $i \in S^t\subseteq g(\theta')$. If $i$ reports $v_i'' > v_i$, it is still selected in stage $t$ and $X^t =\hat{X}^t$, where $\hat{X}^t$ denotes the cost share in stage $t$ given $\theta''$. Hence, $u_i(\theta'')=v_i-\hat{X}^t = v_i-X^t = u_i(\theta')$. If $i$ reports $v_i'' < v_i$, there are two possibilities.
	    \begin{itemize}
	        \item $ X^t\leq v_i'' < v_i$. Then it is still selected and $u_i(\theta'')= u_i(\theta')$.
	        \item $ v_i'' < X^t$. Then it is not selected and $u_i(\theta'')=0\leq u_i(\theta')$.
	    \end{itemize}
	    
	    \item $i \notin g(\theta')$. If $i$ reports $v_i'' < v_i$, then it is still not selected and the utility does not change. If $i$ reports $v_i'' > v_i$, there are two possibilities.
	    \begin{itemize}
	        \item $i \notin g(\theta'')$. Then the utility does not change.
	        \item $i \in S^t\subseteq g(\theta'')$. Then $ x_i(\theta')=\frac{C(S^t)}{|S^t|}$. Since node $i$ cannot be selected given $\theta'$, we have $v_i<\frac{C(S^t)}{|S^t|}$. Hence, $ u_i(\theta'')=v_i-\frac{C(S^t)}{|S^t|}<0= u_i(\theta') $. 
	    \end{itemize}
	\end{enumerate}
	Therefore, $u_i(\theta')\geq u_i(\theta'')$, i.e., $i$'s utility is maximized when $i$ reports its valuation truthfully. 
	\end{proof}

\begin{theorem}
\label{BB_2}
The repeated selection mechanism satisfies budget balance.
\end{theorem}
\begin{proof}
Given $\theta' \in \Theta$, in stage $t$, the sum of all nodes' cost share in $S^t$ equals the total cost of connecting all the nodes in $S^t$, i.e., $X^t\cdot |S^t|=C(S^t)$. Then for all the stages, the sum of all selected nodes' cost share in $g(\theta')$ equals the total cost of connecting all nodes in $g(\theta')$, i.e., $\sum_{t}X^t\cdot |S^t|=\sum_{t}C(S^t)=\sum_{(i,j) \in f(\theta')}c_{(i,j)}$. Hence, the mechanism satisfies budget balance.  
\end{proof}

\begin{theorem}
\label{fea}
The repeated selection mechanism satisfies feasibility.
\end{theorem}
\begin{proof}
Given a report profile $\theta' \in \Theta$, for each node $i \in V\setminus g(\theta')$, we have $x_i(\theta')=0\leq v_i'$. For each node $i \in g(\theta')$, by the proposed mechanism, $x_i(\theta')=X^t\leq v_i'$ for the stage $t$. So we have $x_i(\theta')\leq v_i'$. Therefore, the mechanism satisfies feasibility.  
\end{proof}

\begin{theorem}
\label{}
The repeated selection mechanism satisfies individual rationality.
\end{theorem}
\begin{proof}
Given a report profile $\theta' \in \Theta$, for each node $i \in V\setminus g(\theta')$, we have $u_i(\theta')=0$. For each node $i \in g(\theta')$, by the proposed mechanism, we have $x_i(\theta')\leq v_i'$. By Theorem~\ref{IC_2}, we have $v_i'=v_i$. Hence, we have $u_i(\theta')=v_i -x_i(\theta')\geq 0$. So the mechanism satisfies individual rationality. 
\end{proof}

\begin{theorem}
\label{}
The repeated selection mechanism satisfies positiveness.
\end{theorem}
\begin{proof}
Given a report profile $\theta' \in \Theta$, for each node $i \in V\setminus g(\theta')$, we have $x_i(\theta')=0$. For each node $i \in g(\theta')$, without loss of generality, we assume that it is selected in stage $t$. Obviously, according to the proposed mechanism, its cost share $X^t(\theta')$ is non-negative. Therefore, the mechanism satisfies positiveness.
\end{proof}

\begin{theorem}
The repeated selection mechanism satisfies symmetry.  
\end{theorem}

\begin{proof}
We need to show that, given $\theta' \in \Theta$ and $i,j \in V$ with $r_i(\theta') \setminus \{j\} = r_j(\theta') \setminus \{i\}$ and $v_i=v_j$, $c_{(i,k)} = c_{(j,k)}$ ($\forall k \in r_i(\theta') \setminus \{j\}$) implies $u_i(\theta')=u_j(\theta')$. 

By the proposed mechanism, nodes $i$ and $j$ are either both selected in the same stage or they are not selected. If they are not selected, we have $u_i(\theta')=u_j(\theta')=0$. Without loss of generality, if they are both selected in stage $t$, by the proposed mechanism, we have $x_i(\theta')=x_j(\theta')=X^t$. Since $v_i=v_j$, we have $u_i(\theta')=u_j(\theta')$. So the mechanism satisfies symmetry. 
\end{proof}
\begin{theorem}
\label{ranking}
The repeated selection mechanism satisfies ranking.  
\end{theorem}

\begin{proof}
We need to show that, given $\theta' \in \Theta$ and $i,j \in V$ with $r_i(\theta') \setminus \{j\} = r_j(\theta') \setminus \{i\} $ and $v_i=v_j$, $c_{(i,k)} \leq c_{(j,k)}$ ($\forall k \in r_i(\theta') \setminus \{j\}$) implies $u_i(\theta')\geq u_j(\theta')$. For nodes $i$ and $j$, there are three cases.
\begin{itemize}
    \item $i,j\notin g(\theta')$. Obviously, we have $u_i(\theta')=u_j(\theta')=0$.
    \item $i \in g(\theta'), j\notin g(\theta')$. By individual rationality, we have $u_i(\theta')\geq 0=u_j(\theta')$.
    \item $i,j\in g(\theta')$. Let $t_i$ and $t_j$ denote the stages where nodes $i$ and $j$ are selected respectively. For any set $S$ with $i,j\notin S$, we have $C(S\cup\{i\}) \leq C(S\cup\{j\})$. So we have $t_i\leq t_j$. Since the later the selected stage is, the higher the cost share will be, we have $X^{t_i}\leq X^{t_j}$. Since $v_i=v_j$, we have $u_i(\theta')-u_j(\theta')=x_j(\theta')-x_i(\theta') = X^{t_j}-X^{t_i}\geq 0$.
\end{itemize}
Therefore, the mechanism satisfies ranking.
\end{proof}


    

\begin{theorem}
The repeated selection mechanism satisfies utility monotonicity.  
\end{theorem}
\begin{proof}
Given $\theta' \in \Theta$ and nodes $i,j \in V$ such that the edge $(i,j) \in E$, there are two cases for node $i$.
\begin{itemize}
    \item $i\notin g(\theta')$. Then $u_i(\theta')=0$. When $c_{(i,j)}$ increases, $i$ cannot be selected and the utility remains unchanged.
    \item $i \in g(\theta')$. Then $u_i(\theta')=v_i-x_i(\theta')$. When $c_{(i,j)}$ increases, let $g^+(\theta')$ denote the set of selected nodes.
    \begin{itemize}
        \item $i \notin g^+(\theta')$. Then its utility is $0$. So the utility weakly decreases.
        \item $i \in g^+(\theta')$. According to Theorem~\ref{ranking}, it is easy to show the utility of $i$ weakly decreases.
    \end{itemize}
\end{itemize}
Hence, when $c_{(i,j)}$ increases, the utility of $i$ weakly decreases, i.e., the mechanism satisfies utility monotonicity.
\end{proof}
		  



\section{Conclusions}
\label{con}
In this paper, we study the cost sharing problem under private valuation and connection control on general graphs. We consider two important strategic behaviors of a node (i.e., cutting its adjacent edges and misreporting its valuation). We show that it is impossible for a mechanism to satisfy truthfulness, feasibility, efficiency, and budget balance simultaneously. We also prove that there exists no approximate ratio for efficiency and budget balance. Then we propose two truthful and feasible cost sharing mechanisms that satisfy efficiency or budget balance.

In the future, we try to characterize all possible cost sharing mechanisms that incentivize nodes to share their connections and reveal their valuations. 



\begin{acks}
This work is supported by Science and Technology Commission of Shanghai Municipality (No. 23010503000 and No. 22ZR1442200), and Shanghai Frontiers Science Center of Human-centered Artificial Intelligence (ShangHAI).
\end{acks}



\bibliographystyle{ACM-Reference-Format} 
\bibliography{sample}


\end{document}